\documentclass[12pt,a4paper]{article}

\usepackage{amsmath}
\usepackage{amssymb}
\usepackage{amsthm}
\usepackage{latexsym}
\usepackage[dvipdfmx]{graphicx}
\usepackage{bm}
\usepackage[dvips]{color}

\theoremstyle{plain}
\newtheorem{theorem}{Theorem}
\newtheorem{corollary}[theorem]{Corollary}

\newtheorem{lemma}[theorem]{Lemma}
\newtheorem{claim}{Claim}

\theoremstyle{definition}

\newcommand{\ZZ}{{\mathbb{Z}}}
\newcommand{\RR}{{\mathbb{R}}}
\newcommand{\B}[1]{$b$-{#1}}
\newcommand{\U}{\mathcal{U}}
\newcommand{\C}{\mathcal{C}}

\newcommand{\F}{\mathcal{F}}
\newcommand{\I}{\mathcal{I}}

\DeclareMathOperator{\dom}{{\rm dom}}
\newcommand{\suppp}{\mathrm{supp}\sp{+}}
\newcommand{\suppm}{\mathrm{supp}\sp{-}}

\newcommand{\remove}[1]{}

\usepackage{geometry}
\usepackage{newtxtext,newtxmath}

\title{The $b$-bibranching Problem: \\
TDI System, Packing, and Discrete Convexity}

\author{
Kenjiro Takazawa\thanks{Department of Industrial and Systems Engineering, Faculty of Science and Engineering, 
Hosei University, Tokyo 184-8584, Japan.  
{\tt takazawa@hosei.ac.jp}} 
}

\date{February, 2018}
\begin{document}

\maketitle

\begin{abstract}
In this paper, 
we introduce the $b$-bibranching problem in digraphs, 
which is a common generalization of the bibranching and $b$-branching problems. 
The bibranching problem, 
introduced by Schrijver (1982), 
is a common generalization of the branching and bipartite edge cover problems. 
Previous results on bibranchings include 
polynomial algorithms, 
a linear programming formulation with total dual integrality, 
a packing theorem, 
and 
an M-convex submodular flow formulation. 
The $b$-branching problem, 
recently introduced by Kakimura, Kamiyama, and Takazawa (2018), 
is a generalization of the branching problem admitting higher indegree, 
i.e.,\ 
each vertex $v$ can have indegree at most $b(v)$. 
For $b$-branchings, 
a combinatorial algorithm, 
a linear programming formulation with total dual integrality, 
and a packing theorem for branchings 
are extended. 
A main contribution of this paper is to extend those previous results on 
bibranchings and $b$-branchings
to $b$-bibranchings. 
That is, 
we present a linear programming formulation with total dual integrality, 
a packing theorem, 
and an M-convex submodular flow formulation 
for $b$-bibranchings. 
In particular, 
the linear program and M-convex submodular flow formulations respectively 
imply polynomial algorithms for finding a shortest $b$-bibranching. 
\end{abstract}

\section{Introduction}

In this paper, 
we introduce the \emph{$b$-bibranching problem} in digraphs, 
which is a common generalization of two problems generalizing the branching problem. 
A main contribution of this paper is to provide common extensions of previous theorems on 
these two problems inherited from branchings: 
a linear programming formulation with total dual integrality, 
a packing theorem, 
and an M-convex submodular flow formulation 
for $b$-bibranchings. 

One problem to be generalized is the \emph{bibranching problem}, 
introduced by Schrijver \cite{Sch82} (see also Schrijver \cite{Sch03}). 
The bibranching problem is a common generalization of the branching and bipartite edge cover problems. 
Schrijver \cite{Sch82} proved the total dual integrality of a linear program describing the 
shortest bibranching problem, 
and 
a theorem on packing disjoint bibranchings, 
which extends Edmonds' disjoint branchings theorem \cite{Edm73}. 
The totally dual integral linear program implies the polynomial solvability of the shortest bibranching problem 
via the ellipsoid method, 
and 
it was followed by a faster combinatorial algorithm by Keijsper and Pendavingh \cite{KP98}. 
The integer decomposition property of the bibranching polytope is described in \cite{Sch03}. 
Later, 
Takazawa \cite{Tak12bibr} provided an M-convex submodular flow formulation \cite{Mur99}
for the shortest bibranching problem, 
which also implies a combinatorial polynomial algorithm. 
This formulation is based on 
the discrete convexity of the shortest branchings, 
which is pointed out in  \cite{Tak14} and 
indeed follows from the exchange property of branchings \cite{Sch00}. 
We remark that, 
in the proof for the exchange property, 
Edmonds' disjoint branching theorem \cite{Edm73} plays a key role. 
More recently, 
Murota and Takazawa \cite{MT17} revealed a relation between these two formulations: 
the 
M-convex submodular flow formulation \cite{Tak12bibr} is obtained from 
the linear programming formulation \cite{Sch82} through the Benders decomposition.

The other problem to be generalized is the \emph{$b$-branching problem}, 
recently introduced by Kakimura, Kamiyama, and Takazawa \cite{KKT18}. 
Here, 
$b$ is a positive integer vector on the vertex set of a digraph. 
As is well known, 
a branching is a common independent set of two matroids on the arc set of a digraph: 
one matroid is a partition matroid, 
i.e.,\ 
each vertex $v$ can have indegree at most one; 
and the other matroid is a graphic matroid. 
A $b$-branching is defined as a common independent set of two matroids generalizing these two matroids: 
one matroid imposes that each vertex $v$ can have indegree at most $b(v)$; 
and 
the other matroid is a sparsity matroid defined by $b$ (see Section \ref{SECbb} for precise description). 
We remark that 
a branching is a special case where $b(v)=1$ for every vertex $v$. 
Kakimura, Kamiyama, and Takazawa \cite{KKT18} presented 
a multi-phase greedy algorithm for finding a longest $b$-branching, 
which extends that for branchings \cite{Boc71,CL65,Edm67,Ful74}. 
A theorem on packing disjoint $b$-branchings is also presented in \cite{KKT18}, 
which extends Edmonds' disjoint branchings theorem \cite{Edm73} 
and leads to the integer decomposition property of the $b$-branching polytope. 

In this paper, 
we introduce $b$-bibranchings, 
which provide a common generalization of bibranchings and $b$-branchings. 
We demonstrate that $b$-bibranchings offer a reasonable generalization of 
bibranchings and $b$-branchings by proving extensions of the aforementioned 
results on bibranchings and $b$-branchings. 
We first present a linear programming formulation of the 
shortest $b$-bibranching problem 
and prove its total dual integrality, 
extending those for bibranchings \cite{Sch82} and  $b$-branchings \cite{KKT18}. 
We then prove a theorem on packing disjoint $b$-bibranchings, 
extending those for disjoint bibranchings \cite{Sch82} and 
disjoint $b$-branchings \cite{KKT18}. 
We consequently prove the integer decomposition property of the $b$-bibranching polytope. 
Finally, 
we present an M-convex submodular flow formulation for the shortest $b$-branchings, 
extending that for bibranchings \cite{Tak12bibr}. 

Our proof techniques, 
which might be of theoretical interest, 
are as follows. 
First, 
the total dual integrality is proved by extending the proof for bibranchings in Schrijver \cite{Sch03}. 
Second, 
the proof for 
the packing theorem is based on the supermodular coloring theorem \cite{Sch85}, 
which was used in an alternative proof \cite{Sch85} for the packing theorem for bibranchings. 
Our proof extends Tardos' proof \cite{Tar85} for the supermodular coloring theorem using generalized polymatroids. 
Finally, 
for the M-convex submodular flow formulation, 
we first prove 
an exchange property of $b$-branchings, 
which extends that for branchings \cite{Sch00} 
and 
follows from the theorem for packing $b$-branchings \cite{KKT18}. 
We then establish the discrete convexity of the shortest $b$-branchings. 
Based on this discrete convexity, 
we provide the M-convex submodular flow formulation of the shortest $b$-bibranching problem by extending the arguments in \cite{Tak12bibr}

The organization of this paper is as follows. 
In Section \ref{SECpre}, 
we review previous theorems on 
branchings, 
bibranchings, 
and 
$b$-branchings. 
A formal description of $b$-bibranchings is also presented. 
In Section \ref{SECtdi}, 
we provide a linear programming formulation of the shortest $b$-bibranching problem
and prove its total dual integrality. 
Section \ref{SECbbbpacking} is devoted to establishing a theorem on packing disjoint $b$-bibranchings. 
In Section \ref{SECm}, 
we prove an exchange property of $b$-branchings and 
then establish an M-convex submodular flow formulation for the shortest $b$-bibranching problem. 
Finally, 
we conclude the paper in Section \ref{SECconcl}. 

\section{Preliminaries}
\label{SECpre}

In this section, 
we review previous results on branchings, 
bibranchings, 
and $b$-branchings, 
which will be extended to $b$-bibranchings in subsequent sections. 
At the end of this section, 
we formally define $b$-bibranchings. 

\subsection{Branching}

Throughout this paper, 
we assume that a digraph is loopless. 
Let $D=(V,A)$ be a digraph with vertex set $V$ and arc set $A$. 
We denote an arc $a$ from $u\in V$ to $v \in V$ by $uv$. 
For an arc $a$, 
the initial and terminal vertices are denoted by $\partial ^+ a$ and $\partial^- a$, 
respectively. 
That is, 
if $a = uv$, 
then 
$\partial ^+ a = u$ and $\partial^- a =v$. 
Similarly, 
for an arc subset $B \subseteq A$, 
let 
$\partial^+ B =\bigcup_{a \in B}\{\partial^+ a\}$ 
and 
$\partial^- B =\bigcup_{a \in B}\{\partial^- a\}$. 
For a vertex subset $X \subseteq V$, 
the subgraph of $D$ induced by $X$ is denoted by $D[X]=(X, A[X])$. 
Similarly, for an arc subset $B \subseteq A$, 
the set of arcs in $B$ induced by $X$ is denoted by $B[X]$. 
For vertex subsets $X,Y \subseteq V$, 
let $B[X,Y] = \{a \in B \colon \mbox{$\partial^+a \in X$, $\partial^- a \in Y$}\}$. 
For a vector $x \in \RR^A$ and $B \subseteq A$, 
we denote $x(B)=\sum_{a \in B}x(a)$. 

Let $B \subseteq A$ and 
$\emptyset \neq X \subsetneq V$. 
The set of arcs in $B$ from $X$ to $V \setminus X$ is denoted by $\delta^+_{B}(X)$, 
and 
the set of arcs from $V \setminus X$ to $X$ by $\delta^-_{B}(X)$. 
That is, 
$\delta_{B}^+(X) = \{a \in B \colon \mbox{$\partial^+a \in X$, $\partial^-a \in V \setminus X$}\}$ 
and 
$\delta_{B}^-(X) = \{a \in B \colon \mbox{$\partial^+a \in V \setminus X$, $\partial^-a \in X$}\}$. 
We denote 
$d_B^+(X) = |\delta_B^+(X)|$ 
and 
$d_B^-(X) = |\delta_B^-(X)|$. 
If $B=A$, 
$\delta_{A}^+(X)$ and $\delta_{A}^-(X)$ are often abbreviated as 
$\delta^+(X)$ and $\delta^-(X)$, respectively. 
Also, 
if $X \subseteq V$ is a singleton $\{v\}$, 
then $\delta_{B}^+(\{v\})$, $\delta_{B}^-(\{v\})$,  
$d_{B}^+(\{v\})$, and $d_{B}^-(\{v\})$ are often abbreviated as 
$\delta_{B}^+(v)$, $\delta_{B}^-(v)$, 
$d_{B}^+(v)$, and $d_{B}^-(v)$, 
respectively. 

An arc subset $B\subseteq A$ is called a \emph{branching} 
if $d^-_B(v) \le 1$ for each $v \in V$ and 
the subgraph $(V,B)$ is acyclic. 
An arc subset $B\subseteq A$ is a \emph{cobranching} 
if the reversal of the arcs in $B$ is a branching. 
For a branching $B \subseteq A$, 
define the \emph{root set} $R(B)$ of $B$ by $R(B) = V \setminus \partial^- B$. 
For a cobranching $B$, 
its counterpart $R^*(B)$ is defined by 
$R^*(B) = V \setminus \partial^+ B$. 

As is well known, 
the branchings in a digraph form a special case of matroid intersection. 
Indeed, 
an arc subset is a branching if and only if 
it is a common independent set of a partition matroid $(A, \I_1)$ and 
a graphic matroid $(A,\I_2)$, 
where 
\begin{align}
\label{EQbdeg}
{}&{}\I_1=\{B \subseteq A \colon \mbox{$d^-_B(v) \le 1$ ($v \in V$)}\}, \\
\label{EQbsp}
{}&{}\I_2=\{B \subseteq A \colon \mbox{$|B[X]| \le |X|-1$ ($\emptyset \neq X \subseteq  V)$}\}.
\end{align}

In 
the \emph{longest branching problem}, 
given a digraph $D=(V,A)$ and 
arc weights $w \in \RR\sp{A}$, 
we are asked to find a branching $B$ maximizing $w(B)$. 
The longest branching problem is endowed with a linear programming formulation with 
total dual integrality, 
which is a special case of that for matroid intersection. 
That is, 
the following linear program in variable $x \in \RR^A$ 
is a linear relaxation of the 
longest branching problem, 
where the system \eqref{EQbrpoly1}--\eqref{EQbrpoly3} is a 
linear relaxation of 
matroid constraints \eqref{EQbdeg} and \eqref{EQbsp}, 
and thus it is box-TDI. 
\begin{alignat}{3}
&{}\mbox{maximize}& \quad {}&{}\sum_{a \in A}w(a) x(a) &&\\
\label{EQbrpoly1}
&{}\mbox{subject to}& \quad &{}x(\delta^-(v)) \le 1 \quad {}&&{}(v \in V), \\
\label{EQbrpoly2}
&&&{}x(A[X]) \le |X|-1 \quad    {}&&{}(\emptyset \neq X \subseteq V), \\
\label{EQbrpoly3}
&&&{} x(a) \ge 0               {}&&{}(a\in A). 
\end{alignat}

\begin{theorem}[see \cite{Sch03}]
The linear system \eqref{EQbrpoly1}--\eqref{EQbrpoly3} is box-TDI\@. 
In particular, 
the linear system \eqref{EQbrpoly1}--\eqref{EQbrpoly3} determines the branching polytope. 
\end{theorem}

A theorem for packing disjoint branchings is due to Edmonds \cite{Edm73}. 
For a positive integer $k$, 
let $[k]$ denote the set $\{1,\ldots, k\}$. 

\begin{theorem}[Edmonds \cite{Edm73}]
\label{THMbrpacking}
Let $D=(V,A)$ be a digraph and $k$ be a positive integer. 
For subsets $R_1,\ldots, R_k$  of $V$, 
there exist disjoint branchings $B_1,\ldots, B_k$ such that 
$R(B_j)=R_j$ for each $j \in [k]$ if and only if 
\begin{align}
\notag
d^-_A(X) \ge |\{ j \in [k] \colon R_j \cap X = \emptyset\}| \quad (\emptyset \neq X \subseteq V). 
\end{align}
\end{theorem}

\subsection{Bibranching}

\subsubsection{Definition}

Let $D=(V,A)$ be a digraph, 
where 
$V$ is partitioned into two nonempty subsets $S$ and $T$. 
That is, 
$\emptyset \neq S \subsetneq V$ 
and 
$T = V \setminus S$. 
Schrijver \cite{Sch82} defined that 
an arc subset $B \subseteq A$ is a \emph{bibranching} if 
it satisfies the following two properties: 
\begin{align}
\label{EQbibrdef1}
&{}\mbox{every vertex $v \in T$ is reachable from some vertex in $S$ in the subgraph $(V,B)$}, \\
\label{EQbibrdef2}
&{}\mbox{every vertex $u \in S$ reaches some vertex in $T$ in the subgraph $(V,B)$}. 
\end{align}
Without loss of generality, 
we assume that $D$ does not have an arc from $T$ to $S$. 

Observe that bibranchings offer a common generalization of branchings and bipartite edge covers. 
If $S$ is a singleton $\{s\}$, 
then 
an inclusion-wise minimal bibranching is a branching $B$ with $R(B) = \{s\}$. 
If 
$A[S] = A[T]=\emptyset$, 
then 
the digraph $D$ is bipartite and 
a bibranching is an edge cover in $D$.

In the \emph{shortest bibranching problem}, 
we are given nonnegative arc weights $w \in \RR_+^A$ 
and asked to find a bibranching $B$ minimizing $w(B)$.

An alternative perspective on bibranchings is described in Murota and Takazawa \cite{MT17}:  
an arc subset $B \subseteq A$ is a bibranching if 
\begin{align}
\label{EQMT1}
&\mbox{$d_B^-(v) \ge 1$ for each $v \in T$}, \\
&\mbox{$d_B^+(u) \ge 1$ for each $u \in S$}, \\
&\mbox{$B[T]$ is a branching},\\ 
\label{EQMT4}
&\mbox{$B[S]$ is a cobranching}.   
\end{align}
Although these two definitions slightly differ, 
there would be no confusion 
in considering the shortest bibranching problem and packing disjoint bibranchings.

\subsubsection{Totally dual integral formulation}

An arc subset $C \subseteq A$ is called a \emph{bicut} if 
$C = \delta^-(U)$ for some $U \subseteq V$ with 
$\emptyset \neq U \subseteq T$ or $T \subseteq U \subsetneq V$. 
It is clear that the characteristic vector of a bibranching satisfies 
the following linear system in variable $x \in \RR\sp{A}$: 
\begin{alignat}{2}
\label{EQbibrpoly1}
&{}x(C) \ge 1 \quad{}&&{} \mbox{for each bicut $C$}, \\
\label{EQbibrpoly2}
&{} x(a) \ge 0 {}&&{} \mbox{for each $a \in A$}. 
\end{alignat}
Indeed, 
Schrijver \cite{Sch82} proved that the linear system \eqref{EQbibrpoly1}--\eqref{EQbibrpoly2} is box-TDI. 
\begin{theorem}[Schrijver \cite{Sch82}; see also \cite{Sch03}]
\label{THMbibrTDI}
The linear system \eqref{EQbibrpoly1}--\eqref{EQbibrpoly2} is box-TDI. 
\end{theorem}

\subsubsection{Packing disjoint bibranchings and supermodular coloring}

A theorem on packing disjoint bibranchings is also due to Schrijver \cite{Sch82}. 

\begin{theorem}[Schrijver \cite{Sch82}]
\label{THMbibrpacking}
Let $D=(V,A)$ be a digraph and 
$\{S,T\}$ be a partition of $V$, where $S,T \neq \emptyset$. 
Then, 
the maximum number of disjoint bibranchings in $D$ is equal to 
the minimum size of a bicut. 
\end{theorem}

Note that Theorem \ref{THMbibrpacking} is an extension of 
a special case of Theorem \ref{THMbrpacking}, 
where $R_j=\{s\}$ ($j \in [k]$) for a specified vertex $s \in V$. 

Schrijver \cite{Sch85} presented a proof for Theorem \ref{THMbibrpacking} using the 
\emph{supermodular coloring theorem}, 
which is described as follows.  
For recent progress on supermodular coloring, 
the readers are referred to 
\cite{IY17,Yok18}.

Let $H$ be a finite set. 
A set family $\mathcal{C} \subseteq 2^H$ is an \emph{intersecting family} if, 
for all $X,Y \in \C$ with $X \cap Y \neq \emptyset$, 
it holds that 
$X\cup Y, X\cap Y \in \C$. 
For an intersecting family $\C \subseteq 2^H$, 
a function $g:\C \to \RR$ is called \emph{intersecting supermodular} if 
$g(X) + g(Y) \le g(X \cup Y) + g(X \cap Y)$ holds for all $X,Y \in \C$ with $X \cap Y \neq \emptyset$. 
A function $f:\C \to \RR$ is called \emph{intersecting submodular} if 
$-f$ is supermodular.

\begin{theorem}[Schrijver \cite{Sch85}]
\label{THMsupcol}
Let $\C_1,\C_2 \subseteq 2^H$ be intersecting families, 
$g_1\colon \C_1 \to \RR$ and 
$g_2\colon \C_2 \to \RR$ be intersecting supermodular functions, 
and $k$ be a positive integer.  
Then, 
$H$ can be partitioned into 
$k$ classes $H_1,\ldots,H_k$ 
such that 
$g_i(C) \le |\{j\in [k] \colon H_j \cap C \neq \emptyset\}|$ 
for each $C \in \C_i$ and $i=1,2$ 
if and only if 
$g_i(C) \le \min \{ k, |C| \}$ for each $i=1,2$ and each $C \in \C_i$. 
\end{theorem}

In Section \ref{SECbbbpacking}, 
we prove a theorem on packing disjoint $b$-bibranchings (Theorem \ref{THMbbbpacking}), 
which extends Theorem \ref{THMbibrpacking}, 
by extending Tardos' proof for Theorem \ref{THMsupcol} using generalized polymatroids \cite{Fra84}. 
A generalized polymatroid is a polyhedron defined 
by an intersecting supermodular function 
and an intersecting submodular function with a certain property. 
Here we omit the definition, 
but 
show basic properties of generalized polymatroids 
used in the subsequent sections. 

Let $P \subseteq \RR^H$ be a polyhedron. 
Let $a_1,a_2\in H$ and $\tilde{a}$ be an element not belonging to $H$. 
Denote $\tilde{H} = (H \setminus \{a_1,a_2\}) \cup \{\tilde{a}\}$. 
The \emph{aggregation} of $P$ at $a_1,a_2\in H$ is a polyhedron $\tilde{P} \in \RR^{\tilde{V}}$ defined by 
\begin{align*}
\tilde{P} = \{ (x_0, x(a_1)+x(a_2)) \colon (x_0, x(a_1), x(a_2)) \in P \}, 
\end{align*}
where $x_0 \in \RR^{H\setminus\{a_1,a_2\}}$. 
Let $a \in H$, 
and $a',a''$ be elements not belonging to $H$. 
Denote $H'= (H \setminus \{a\}) \cup \{a',a''\}$. 
The \emph{splitting} of $P$ at $a \in H$ is a polyhedron $P'\subseteq \RR^{H'}$ 
defined by 
\begin{align*}
P' = \{ (x_0; x(a'), x(a'')) \colon (x_0;x(a')+x(a'')) \in P  \}, 
\end{align*}
where $x_0 \in \RR^{H\setminus\{a\}}$. 

\begin{theorem}[See \cite{FT88,Fuj05,Mur03,Sch03}]
\label{THgp}
Generalized polymatroids have the following properties. 
\begin{enumerate}
\item 
\label{ENUinteger}
A generalized polymatroid is integer if and only if 
it is determined by a pair of an intersecting submodular function and an intersecting supermodular function 
which are integer. 
\item 
\label{ENUintersection}
The intersection of two integer generalized polymatroids is an integer polyhedron. 
\item
\label{ENUclosed}
Generalized polymatroids are closed under the operations of 
splitting, 
aggregation, 
and
intersection with a box. 
\end{enumerate}
\end{theorem}

\subsubsection{M-convex submodular flow formulation}

We finally review the $\mathrm{M}^\natural$-convex submodular flow formulation for the shortest bibranching problem \cite{Tak12bibr}. 
We begin some definitions. 
Let $\overline{\ZZ}$ denote $\ZZ \cup \{ + \infty\}$. 
Let $V$ be a finite set. 
For a vector $x \in \RR^V$, 
define 
$\suppp(x) = \{v \in V \colon x(v)>0\}$ 
and 
$\suppm(x) = \{v \in V \colon x(v)<0\}$.  
For $v \in V$, 
let $\chi_v$ denote a vector in $\ZZ^V$ defined by 
$\chi_v(v) =1$ and 
$\chi_v(v') = 0$ for each $v' \in V \setminus \{v\}$. 
A function $f\colon \ZZ^V \to \overline{\ZZ}$ is an \emph{$\mathrm{M}^\natural$-convex function} \cite{Mur03,MS99,MS18} if 
it satisfies the following property: 
\begin{quote}
For each $x,y \in \ZZ^V$ and $u \in \suppp(x-y)$, 
\begin{align}
\label{EQM1}
f(x) + f(y) \ge f(x - \chi_u) + f(y + \chi_u), 
\end{align}
or there exists $v \in \suppm(x-y)$ such that 
\begin{align}
\label{EQM2}
f(x) + f(y) \ge f(x - \chi_u+\chi_v) + f(y + \chi_u-\chi_v).
\end{align}
\end{quote}
The effective demain $\dom f$ of $f$ is defined by 
$\dom f = \{x \in \ZZ^V \colon f(x) < + \infty\}$. 

Let $D=(V,A)$ be a digraph and 
$w \in \RR^A_+$ represent the arc weights. 
Let $\underline{c},\overline{c} \in \RR^A$ be vectors on $A$ such that 
$\underline{c}(a) \le \overline{c}(a)$ for each $a \in A$. 
For $\xi \in \RR\sp{A}$, 
define $\partial^+\xi,\partial^-\xi,\partial \xi \in \RR\sp{V}$ by 
\begin{align*}
&{}\partial^+ \xi (v) = 
\sum_{a \in \delta^+(v)}\xi(a) \quad (v \in V), \\
&{}\partial^- \xi (v) = 
\sum_{a \in \delta^-(v)}\xi(a) \quad (v \in V), \\
&{}\partial \xi (v) = 
\partial^+\xi(v) - \partial^-\xi(v) \quad(v \in V). 
\end{align*}
Let $f: \ZZ^V \to \overline{\ZZ}$ be an $\mathrm{M}^\natural$-convex function. 
Now the following problem 
in variable $\xi \in \ZZ^A$ 
is called the \emph{$\mathrm{M}^\natural$-convex submodular flow problem} \cite{Mur99}:
\begin{alignat}{2}
&{}\mbox{minimize} \quad 	{}&&{}\sum_{a \in A}w(a)\xi(a) + f(\partial \xi) \\
&{}\mbox{subject to} \quad 	{}&&{}
\underline{c}(a) \le \xi(a) \le \overline{c}(a) \quad \mbox{for each $a \in A$}, \\
&{} {}&&{}\partial\xi \in \dom f.
\end{alignat}

The $\mathrm{M}^\natural$-convex submodular flow formulation for the 
shortest bibranching problem \cite{Tak12bibr} is obtained as follows. 
Define a function $f_T \colon \ZZ^T \to \overline{\ZZ}$  
in the following manner. 
First, 
the effective domain $\dom f_T$ is defined by 
\begin{align}
\label{EQdomfT}
\dom f_T = \{x \in \ZZ_+^T \colon \mbox{$D[T]$ has a branching $B$ with $x \ge \chi_{R(B)}$}\}. 
\end{align}
Then, 
for $x \in \ZZ^T$, 
the function value $f_T(x)$ is defined by 
\begin{align}
\label{EQvalfT}
f_T(x) = 
\begin{cases}
\min\{ w(B) \colon \mbox{$B$ is a branching in $D[T]$, $x \ge \chi_{R(B)}$} \} & (x \in \dom f_T), \\
+\infty & (x \not \in \dom f_T). 
\end{cases}
\end{align}
Similarly, 
define a function $f_S \colon \ZZ^S \to \overline{\ZZ}$ by 
\begin{align*}
&{}\dom f_S = \{x \in \ZZ_+^S \colon \mbox{$D[S]$ has a cobranching $B$ with $x \ge \chi_{R^*(B)}$}\}, \\
&{}f_S(x) = 
\begin{cases}
\min\{ w(B^*) \colon \mbox{$B^*$ is a cobranching in $D[S]$, $x \ge \chi_{R^*(B^*)}$} \} & (x \in \dom f_S), \\
+\infty & (x \not \in \dom f_S). 
\end{cases}
\end{align*}
Now the $\mathrm{M}^\natural$-convexity of $f_T$ and $f_S$ is derived from the exchange property of branchings \cite{Sch00}. 
\begin{theorem}[\cite{Tak12bibr}, see also \cite{MT17}]
\label{THMMconv}
The functions $f_T$ and $f_S$ are $\mathrm{M}^\natural$-convex. 
\end{theorem}

Based on the perspective on bibranchings by Murota and Takazawa \cite{MT17}, 
we can describe 
the shortest bibranching problem as 
the following nonlinear minimization problem in variable $\xi\in \ZZ\sp{A[S,T]}$: 
\begin{alignat}{2}
\label{EQbbMforma}
&{}\mbox{minimize} \quad 	{}&&{}\sum_{a \in A[S,T]}w(a)\xi(a) + f_S(\partial^+ \xi) + f_T(\partial^- \xi) \\
&{}\mbox{subject to} \quad 	{}&&{}
0 \le \xi(a) \le 1 \quad \mbox{for each $a \in A[S,T]$}, \\
&{} {}&&{}\partial^+\xi \in \dom f_S , \\
\label{EQbbMformz}
&{} {}&&{}\partial^-\xi \in \dom f_T .
\end{alignat}

For $x \in \RR^V$ and $U \subseteq V$, 
denote the restriction of $x$ to $U$ by $x|_U$. 
It directly follows from Theorem \ref{THMMconv} that 
a function $f : \ZZ^{A[S,T]} \to \overline{\ZZ}$ defined by 
$f(x)=f_S(x|_S) + f_T(x|_T)$ ($x \in \ZZ^{V}$) 
is an $\mathrm{M}^\natural$-convex function. 
Therefore, 
the minimization problem \eqref{EQbbMforma}--\eqref{EQbbMformz} is an instance of the $\mathrm{M}^\natural$-convex submodular flow problem.

\subsection{$b$-branching}
\label{SECbb}

Let $D=(V,A)$ be a digraph and $b \in \ZZ_{++}^V$ be a positive integer vector on $V$. 
An arc subset $B \subseteq A$ is a \emph{$b$-branching} \cite{KKT18} if 
\begin{alignat}{2}
\label{EQbbdeg}
&{}d_B^-(v) \le b(v) \quad {}&&{}(v \in V), \\
\label{EQbbsp}
&{}|B[X]| \le b(X) - 1 \quad {}&&{}(\emptyset \neq X \subseteq V). 
\end{alignat}
It is clear that, 
in the case $b(v)=1$ for each $v \in V$, 
a $b$-branching is exactly a branching: 
\eqref{EQbbdeg} and \eqref{EQbbsp} correspond to 
\eqref{EQbdeg} and \eqref{EQbsp}, 
respectively. 

Also, 
observe that 
\eqref{EQbbdeg} defines an independenet set family of a matroid. 
Moreover, 
\eqref{EQbbsp} as well defines a matroid, 
called a \emph{sparsity matroid} or a \emph{count matroid} (see, e.g.,\ \cite{Fra11}). 
Therefore, 
a $b$-branching is a special case of matroid intersection. 

This observation leads to the fact that  
the following linear system, 
in variable $x \in \RR^A$, determines the matroid intersection polytope and 
thus box-TDI: 
\begin{alignat}{2}
\label{EQbbpoly1}
&{}x(\delta^-(v)) \le b(v) \quad {}&&{}(v \in V), \\
\label{EQbbpoly2}
&{}x(A[X]) \le b(X)-1 \quad    {}&&{}(\emptyset \neq X \subseteq V), \\
\label{EQbbpoly3}
&{} 0\le x(a) \le 1               {}&&{}(a\in A). 
\end{alignat}
\begin{theorem}[\cite{KKT18}]
\label{THMbbTDI}
The linear system \eqref{EQbbpoly1}--\eqref{EQbbpoly3} is 
box-TDI. 
In particular, 
the linear system \eqref{EQbbpoly1}--\eqref{EQbbpoly3} determines 
the $b$-branching polytope. 

\end{theorem}

What is more, 
$b$-branchings inherit several good properties of branchings. 
In \cite{KKT18}, 
a multi-phase greedy algorithm for finding a longest $b$-branching 
and 
a theorem on packing disjoint $b$-branchings are presented. 
The former is an extension of that for finding a longest branching \cite{Boc71,CL65,Edm67,Ful74}. 
The latter is an extension of that for packing disjoint branchings (Theorem \ref{THMbrpacking}) and 
is described as follows.

\begin{theorem}[\cite{KKT18}]
\label{THMbbpacking}
Let $D=(V,A)$ be a digraph, 
$b \in \ZZ_{++}^V$ be a positive integer vector on $V$, 
and 
$k$ be a positive integer. 
For $j \in [k]$, 
let $b_j \in \ZZ_{+}^V$ be a vector such that 
$b_j(v) \le b(v)$ for every $v \in V$ and 
$b_j \neq b$. 
Then, 
$D$ has disjoint $b$-branchings $B_1,\ldots, B_k$ such that 
$d_{B_j}^- = b_j$
if and only if 
the following two conditions are satisfied: 
\begin{alignat}{2}
\label{EQpackingdeg}
{}&{}d_A^-(v) \ge \sum_{j=1}^k b_j(v)\quad {}&&{}(v \in V),\\
\label{EQpackingcut}
{}&{}d_A^-(X) \ge |\{ j \in [k] \colon b_j(X) = b(X) \neq 0\}| \quad{}&&{} (\emptyset \neq X \subseteq V). 
\end{alignat}
\end{theorem}

\subsection{Definition of $b$-bibranching}

We finally define \emph{$b$-bibranchings}, 
the central concept in this paper. 
Let $D=(V,A)$ be a digraph, 
and let $V$ be partitioned into two nonempty subsets $S$ and $T$. 
Let $b  \in \ZZ_{++}^V$ be a positive integer vector on the vertex set $V$. 
An arc subset $B \subseteq A$ is a \emph{$b$-bibranching} if 
it satisfies the following four properties: 
\begin{align}
\label{EQbbbdef1}
&{}\mbox{every vertex $v \in T$ is reachable from some vertex in $S$ in the subgraph $(V,B)$}, \\
\label{EQbbbdef2}
&{}\mbox{every vertex $u \in S$ reaches some vertex in $T$ in the subgraph $(V,B)$}, \\
\label{EQbbbdef3}
&{}\mbox{$d^-_B(v) \ge b(v)$ for every $v \in T$}, \\
\label{EQbbbdef4}
&{}\mbox{$d^+_B(u) \ge b(u)$ for every $u \in S$}. 
\end{align}

The aforementioned special cases of $b$-bibranchings, 
i.e.,\ 
branchings, bibranchings, and $b$-branchings, 
are obtained as follows. 
If we assume that 
(a) $b(v)=1$ for every $v \in V$ 
and 
(b) $S$ is a singleton $\{s\}$, 
then 
an inclusion-wise minimal $b$-bibranching $B$ is exactly a branching with $d_B^-(v) = 1$ for each $v \in V \setminus \{s\}$ (an $s$-arborescence). 
If we only have Assumption (a), 
a $b$-bibranching is exactly a bibranching 
defined by Schrijver \cite{Sch82}. 
If we only have Assumption (b), 
an inclusion-wise minimal $b$-bibranching $B$ is exactly a $b$-branching \cite{KKT18} 
with $d_B^-(s)=0$ and $d_B^-(v) = b(v)$ for each $v \in V \setminus \{s\}$. 

In view of the definition of 
bibranchings by \eqref{EQMT1}--\eqref{EQMT4}, 
an 
alternative description of $b$-bibranchings is 
as follows. 
Call an arc subset $B \subseteq A$ a 
\emph{$b$-cobranching} if 
the reversal of the arcs in $B$ is a $b$-branching. 
Then, 
an arc subset $B \subseteq A$ is a $b$-bibranching if 
$B[T]$ is a $b|_T$-branching in $D[T]$ and 
$B[S]$ is a $b|_S$-cobranching in $D[S]$, 
as well as 
\eqref{EQbbbdef3} and \eqref{EQbbbdef4}. 

In the sequel, 
we present extensions of the aforementioned results on 
branchings, 
bibranchings, 
and $b$-branchings 
to $b$-bibranchings. 

\section{TDI system for $b$-bibranchings}
\label{SECtdi}

In this section, 
we present a linear programming formulation for the shortest $b$-bibranching problem, 
and prove its total dual integrality. 
This is a common extension of that for bibranchings (Theorem \ref{THMbibrTDI}) and 
that for $b$-branchings (Theorem \ref{THMbbTDI}). 
Our proof is based on that for bibranchings by Schrijver \cite{Sch03}. 

Let $D=(V,A)$ be a digraph, 
$w \in \RR_+^A$ be a vector representing the arc weights, 
$\{S,T\}$ be a partition of $V$, where $S,T \neq \emptyset$, 
and 
$b\in \ZZ_{++}^V$ be a positive integer vector on $V$. 
The following linear program in variable $x \in \RR\sp{A}$ is a relaxation of the 
shortest $b$-bibranching problem: 
\begin{alignat}{3}
\label{EQlpobj}
&{}\mbox{minimize}\quad{}&&{} \sum_{a \in A}w(a)x(a) &&\\
&{}\mbox{subject to} {}&
\label{EQpoly1}
&{}x(\delta^-(v)) \ge b(v) \quad{}&&{} \mbox{for each $v \in T$}, \\
&&&{}x(\delta^+(v)) \ge b(v) {}&&{} \mbox{for each $v \in S$}, \\
\label{EQpoly3}
&&&{}x(C) \ge 1 {}&&{} \mbox{for each bicut $C$}, \\
\label{EQpoly4}
&&&{} x(a) \ge 0 {}&&{} \mbox{for each $a \in A$}. 
\end{alignat}
Note that an integer feasible solution $x$ for this linear program can have $x(a)\ge 2$. 
In such a case, 
$x$ is not the characteristic vector of a $b$-bibranching. 
However, 
we prove that this linear program is indeed box-TDI (Theorem \ref{THMbbbtdi}), 
and then obtain a linear description of the $b$-bibranching polytope by 
taking the intersection with a box $[0,1]^{A}$ (Corollary \ref{CORbbbpoly}).

\begin{theorem}
\label{THMbbbtdi}
The linear system \eqref{EQpoly1}--\eqref{EQpoly4} is box-TDI. 
\end{theorem}

\begin{proof}
Define $\U \subseteq 2^V$ by 
\begin{align*}
&{}\U =  \{\{v\} \colon v \in V\} \cup \U', &
&{}\U' = \{U \subseteq T \colon |U| \ge 2\} \cup \{U \supseteq T \colon |V \setminus U| \ge 2\}. 
\end{align*}
Now consider the following dual linear program of \eqref{EQlpobj}--\eqref{EQpoly4} in variable $y \in \RR\sp{\U}$: 
\begin{alignat}{3}
\label{EQdual1}
&{}\mbox{maximize} \quad 	{}&&{}\sum_{v \in V}b(v)y(v) + \sum_{U \in \U'}y(U) &&\\
\label{EQdual2}
&{}\mbox{subject to} \quad 	&&{}{}y(\partial^-a)+\sum_{U \in \U', a \in \delta^-U}y(U) \le w(a) \quad {}&&{}\mbox{for each $a \in A[T]$},\\
&{} {}&&{}{}y(\partial^+a)+\sum_{U \in \U', a \in \delta^-U}y(U) \le w(a) \quad {}&&{}\mbox{for each $a \in A[S]$},\\
&{} {}&&{}{}y(\partial^-a)+y(\partial^+a)+\sum_{U \in \U', a \in \delta^-U}y(U) \le w(a) \quad {}&&{}\mbox{for each $a \in A[S,T]$},\\
\label{EQdual3}
&{} {}&&{}y(U) \ge 0 \quad {}&&{}\mbox{for each $U \in \U$}. 
\end{alignat}

Let $y^* \in \RR\sp{\U}$ be an optimal solution for the linear program \eqref{EQdual1}--\eqref{EQdual3} 
minimizing 
$\sum_{U \in \U} y(U) \cdot |U| \cdot |V \setminus U|$. 
We prove that 
the collection of $U \in \U$ such that $y^*(U)>0$ is cross-free, 
i.e., 
there exists no pair of $X,Y \in \U$ ($X \neq Y$) 
such that 
$y^*(X),y^*(Y) >0$ 
and 
the four sets  
$X \setminus Y$, 
$Y \setminus X$, 
$X \cap Y$, 
and 
$V \setminus (X \cup Y)$ 
are nonempty. 
\begin{claim}
\label{CLcrossfree}
The vertex subset family $\F^*=\{U \in \U \colon y^*(U) >0\}$ is cross-free. 
\end{claim}
\begin{proof}[Proof of Claim \ref{CLcrossfree}.]
Assume to the contrary that $\F^*$ is not cross-free because 
$X,Y \in \F^*$ violate the condition. 
Let $\alpha = \min\{y^*(X),y^*(Y)\}$ and 
define $y' \in \RR\sp{\U}$ by 
\begin{align}
y'(U) = 
\begin{cases}
y^*(U) - \alpha & (U = X,Y), \\
y^*(U) + \alpha & (U = X\cup Y, X\cap Y), \\
y^*(U) & (\mbox{otherwise}). \\
\end{cases}
\end{align}
Then, 
it is straightforward to see that 
$y'$ satisfies \eqref{EQdual2} and \eqref{EQdual3}. 
It is also not difficult to see that 
the value of \eqref{EQdual1} when $y=y'$ is at least the value of \eqref{EQdual1} when $y=y^*$. 
Indeed, 
if $X \cap Y = \{v\}$ for some $v \in V$, 
the value increases by $(b(v) - 1)\alpha$, 
and otherwise 
the value does not change. 
Therefore, 
$y'$ is also an optimal solution for the linear program \eqref{EQdual1}--\eqref{EQdual3}. 
Moreover,  
it holds that 
$\sum_{U \in \U} y'(U) \cdot |U| \cdot |V \setminus U| < \sum_{U \in \U} y^*(U) \cdot |U| \cdot |V \setminus U|$. 
This contradicts the minimality of $y^*$, 
and thus we conclude that $\F^*$ is cross-free. 
\end{proof}

From Claim \ref{CLcrossfree}, 
it follows that 
the $\F^* \times A$ matrix $M$ defined below is a network matrix \cite[Theorem 54.8]{Sch03}:  
\begin{align}
M_{U,a} = 
\begin{cases}
1 & (a \in \delta^-(U)), \\
0 & (\mbox{otherwise}) 
\end{cases}
\quad (U \in \F^*, a \in A). 
\end{align}
Since a network matrix is totally unimodular \cite{Tut65}, 
we obtain that $M$ is totally unimodular. 
It then follows that 
the system \eqref{EQpoly1}--\eqref{EQpoly4} is box-TDI \cite[Theorem 5.35]{Sch03}.\@ 
\end{proof}

The following corollary is a direct consequence of Theorem \ref{THMbbbtdi}. 

\begin{corollary}
\label{CORbbbpoly}
The linear system defined by 
\eqref{EQpoly1}--\eqref{EQpoly4} and 
\begin{align}
\label{EQpoly5}
&{}  x(a) \le 1 \quad \mbox{for each $a \in A$}. 
\end{align}
is totally dual integral. 
In particular, 
the $b$-bibranching polytope is determined by \eqref{EQpoly1}--\eqref{EQpoly4} and \eqref{EQpoly5}. 
\end{corollary}

From Corollary \ref{CORbbbpoly}, 
it follows that the shortest $b$-bibranching problem can be solved in polynomial time via the ellipsoid method. 

\begin{corollary}
The shortest $b$-bibranching problem can be solved in polynomial time. 
\end{corollary}

\section{Packing disjoint \B{bibranchings}}
\label{SECbbbpacking}

In this section, 
we prove a theorem on packing disjoint $b$-branchings, 
which is a common extension of Theorems \ref{THMbrpacking}, \ref{THMbibrpacking}, and \ref{THMbbpacking}. 
Our proof is an extension of Tardos' proof for the supermodular coloring theorem (Theorem \ref{THMsupcol}) 
using generalized polymatroids. 

\begin{theorem}
\label{THMbbbpacking}
Let $D=(V,A)$ be a digraph, 
$\{S,T\}$ be a partition of $V$, where $S,T \neq \emptyset$, 
and 
$b\in \ZZ_{++}^V$ be a positive integer vector on $V$. 
Then, 
the maximum number of disjoint $b$-bibranchings is equal 
to the minimum of the following three values: 
\begin{align}
\label{EQminT}
&{}\min\left\{ \left\lfloor\frac{d_A^-(v)}{b(v)}\right\rfloor \colon v \in T \right\}; \\
\label{EQminS}
&{}\min\left\{ \left\lfloor\frac{d_A^+(v)}{b(v)}\right\rfloor \colon v \in S \right\}; \\
\label{EQminC}
&{}\min\{ |C| \colon \mbox{$C$ is a bicut}\}.
\end{align}
\end{theorem}

\begin{proof}
It is straightforward to see that 
the maximum number of disjoint $b$-bibranchings is at most the minimum of 
\eqref{EQminT}--\eqref{EQminC}. 
In what follows, 
we prove the opposite inequality. 

Let $k \in \ZZ_+$ be the minimum of \eqref{EQminT}--\eqref{EQminC}. 
Denote $H = A[S,T]$. 
Define $\C_1, \C_2 \subseteq 2^H$ by 
\begin{align*}
&{}\C_1 = \{ \delta^-_H(U) \colon \emptyset \neq U \subseteq T \}, &
&{}\C_2 = \{ \delta^+_H(U) \colon \emptyset \neq U \subseteq S \}.
\end{align*}
Note that 
\begin{align}
\label{EQpartition}
&\delta_H^-(U) = \bigcup_{v\in U} \delta_H^-(v) \quad (\emptyset \neq U \subseteq T), &
&\delta_H^+(U) = \bigcup_{v\in U} \delta_H^+(v) \quad (\emptyset \neq U \subseteq S). 
\end{align}
Define two functions $g_1 \colon \C_1 \to \ZZ$ and $g_2 \colon \C_2 \to \ZZ$ by 
\begin{align}
\label{EQg1}
&{}g_1(C) = \max\{ k - d^-_{A[T]}(U) \colon \mbox{$\emptyset \neq U \subseteq T$}, \mbox{$C = \delta_H^-(U)$}\} \quad (C \in \C_1), \\
&{}g_2(C) = \max\{ k - d^-_{A[S]}(U) \colon \mbox{$\emptyset \neq U \subseteq S$}, \mbox{$C = \delta_H^+(U)$}\} \quad (C \in \C_2). 
\end{align}
We now prepare Claims \ref{CLintersecting}--\ref{CLk} below. 
\begin{claim}
\label{CLintersecting}
The subset families $\C_1$ and $\C_2$ of $H$ are intersecting families. 
\end{claim}
\begin{proof}[Proof of Claim \ref{CLintersecting}]
We prove that $\C_1$ is an intersecting family. 
The argument directly applies to $C_2$ as well. 

Let $X,Y \in \C_1$ satisfy $X \cap Y \neq \emptyset$. 
Since $X,Y \in \C_1$, 
it follows that $X = \delta_H^-(U_X)$ and $Y = \delta_H^-(U_Y)$ for some nonempty sets $U_X,U_Y \subseteq T$.  
Then, 
it directly follows from \eqref{EQpartition} that 
$X \cup Y = \delta_H^-(U_X \cup U_Y)$ and 
$X \cap Y = \delta_H^-(U_X \cap U_Y)$.
Furthermore, 
it follows 
from $X \cap Y \neq \emptyset$ 
that $U_X \cap U_Y \neq \emptyset$. 
We thus obtain that $X \cup Y, X \cap Y \in \C_1$, 
implying that $\C_1$ is an intersecting family. 
\end{proof}

\begin{claim}
\label{CLism}
Functions $g_1$ and $g_2$ are intersecting supermodular. 
\end{claim}

\begin{proof}[Proof of Claim \ref{CLism}]
We prove that $g_1$ is an intersecting supermodular function. 
In the same manner, 
we can prove that $g_2$ is as well an intersecting supermodular function. 

Let $X,Y \in \C_1$ satisfy $X \cap Y \neq \emptyset$, 
and let $U_X, U_Y \subseteq T$ satisfy that 
$U_X, U_Y \neq \emptyset$, 
$X= \delta_H^-(U_X)$, 
$g_1(X) = k - d^-_{A[T]}(U_X)$, 
$Y= \delta_H^-(U_Y)$, 
and  
$g_1(Y) = k - d^-_{A[T]}(U_Y)$. 
Then, 
it follows from 
the definition \eqref{EQg1} of $g_1$ that 
\begin{align}
\label{EQintsup3}
&{}g_1(X\cup Y) \ge k - d_{A[T]}^-(U_X \cup U_Y), 
& 
&{}g_1(X\cap Y) \ge k - d_{A[T]}^-(U_X \cap U_Y). 
\end{align}
By the submodularity of $d_{A[T]}^-$, 
we have that 
\begin{align}
\label{EQintsup5}
d_{A[T]}^-(X) + d_{A[T]}^-(Y) 
\ge 
d_{A[T]}^-(X\cup Y)  + d_{A[T]}^-(X\cap Y) .
\end{align}
By \eqref{EQintsup3} and \eqref{EQintsup5}, 
we obtain 
\begin{align*}
g_1(X\cup Y)+g_1(X\cap Y) 
\ge {}&{}2k - d_{A[T]}^-(U_X \cup U_Y) - d_{A[T]}^-(U_X \cap U_Y) \\
\ge {}&{}2k - d_{A[T]}^-(U_X) - d_{A[T]}^-(U_Y) \\
=   {}&{}g_1(X) + g_1(Y). 
\end{align*}
We thus conclude that $g_1$ is intersecting supermodular. 
\end{proof}

\begin{claim}
\label{CLk}
For $i=1,2$, 
\begin{align}
\label{EQcolorassump}
g_i(C) \le \min\{k, |C|\} \quad \mbox{for each $C \in \C_i$}. 
\end{align}
\end{claim}

\begin{proof}
We show the case $i=1$. 
The other case $i=2$ can be shown in the same manner. 

For $C \in \C_1$, 
it directly follows from the definition \eqref{EQg1} of $g_1$ 
that $g_1(C)\le k$. 
To prove $g_1(C)\le |C|$, 
let $U$ be a nonempty subset of $T$ such that 
$C = \delta_H^-(U)$ and 
$g_1(C) = k - d^-_{A[T]}(U)$. 
Then, 
\begin{align}
\label{EQatmostC}
g_1(C) 
= k - d^-_{A[T]}(U) 
\le d^-_{A}(U) - d^-_{A[T]}(U) 
= d^-_{H}(U) 
= |C|, 
\end{align}
where the inequality in \eqref{EQatmostC} follows from the definition of $k$, 
i.e., 
$k$ is at most \eqref{EQminC}. 
We thus conclude that $g_1(C) \le \min\{k,|C|\}$ for each $C \in \C_1$. 
\end{proof}

Consider the following linear system in variable $x \in \RR^H$: 
\begin{alignat}{2}
\label{EQgpoly1}
{}&{} 0 \le x(a) \le 1 					{}&&{} \quad(a \in H), \\
\label{EQgpoly2}
{}&{} x(C) \le |C| - g_1(C) + 1			{}&&{} \quad(C \in \C_1), \\
\label{EQgpoly3}
{}&{} x(C) \ge 1 						{}&&{} \quad(\mbox{$C \in \C_1$, $g_1(C) = k$}), \\
\label{EQgpoly4}
{}&{} x(\delta_H^-(v)) \le d_A^-(v) - (k-1)b(v)	{}&&{} \quad(v \in T). 
\end{alignat}
Denote the polytope determined by \eqref{EQgpoly1}--\eqref{EQgpoly4} by $P_1 \subseteq \RR^H$. 
Note that, 
since $k$ is at most \eqref{EQminT}, 
the right-hand side 
$d_A^-(v) - (k-1)b(v)$ of \eqref{EQgpoly4} is nonnegative. 
We also define a polytope $P_2 \subseteq \RR^H$
by the following system: 
\begin{alignat}{2}
\label{EQgpoly21}
{}&{} 0 \le x(a) \le 1 					{}&&{} \quad(a \in H), \\
\label{EQgpoly22}
{}&{} x(C) \le |C| - g_2(C) + 1			{}&&{} \quad(C \in \C_2), \\
\label{EQgpoly23}
{}&{} x(C) \ge 1 						{}&&{} \quad(\mbox{$C \in \C_2$, $g_2(C) = k$}), \\
\label{EQgpoly24}
{}&{} x(\delta_H^+(v)) \le d_A^+(v) - (k-1)b(v)	{}&&{} \quad(v \in S). 
\end{alignat}
We now show that $P_1 \cap P_2$ contains an integer vector by Claims \ref{CLgp} and \ref{CLoverk} below. 

\begin{claim}
\label{CLgp}
The polytopes $P_1$ and $P_2$ are generalized polymatroids. 
\end{claim}

\begin{proof}
Here we prove that $P_1$ is a generalized polymatroid. 
In the same manner, 
$P_2$ can be proved to be a generalized polymatroid. 

By following the argument in Schrijver \cite[Theorem 49.14]{Sch03}, 
we obtain from Claims \ref{CLintersecting}--\ref{CLk} that the polytope $P \subseteq \RR^H$ determined by \eqref{EQgpoly1}--\eqref{EQgpoly3} is a generalized polymatroid. 
Here we prove that the addition of the constraint \eqref{EQgpoly4} maintains that 
the determined polytope is a generalized polymatroid. 

Consider the following sequence of 
transformations: 
\begin{align}
\label{EQtrans1}
&{}Q = \{ y \in \RR^T \colon \mbox{$\exists x \in P$, $x(\delta_H^-(v))= y(v)$ for each $v \in T$} \}, \\
\label{EQtrans2}
&{}R = \{ y \in \RR^T \colon \mbox{$y \in Q$, $y(v) \le d_A^-(v) - (k-1)b(v)$ for each $v \in T$} \}, \\
\label{EQtrans3}
&{}P' = \{ x \in \RR^H \colon \mbox{$\exists y \in R$, $y(v) = x(\delta_H^-(v))$ for each $v \in T$} \}. 
\end{align}
It is straightforward to see that $P'=P_1$. 
We complete the proof by showing that the transformations \eqref{EQtrans1}--\eqref{EQtrans3} 
maintain that the polytope is a generalized polymatroid. 

First, 
$Q$ is the aggregation of $P$, 
and hence is a generalized polymatroid (Theorem \ref{THgp}\ref{ENUclosed}). 
Next, 
$R$ is the intersection of $Q$ and a box $[0,d_A^- - (k-1)b]$. 
Hence $R$ is again a generalized polymatroid (Theorem \ref{THgp}\ref{ENUclosed}). 
Finally, $P'$ is the splitting of $R$, 
and hence $P'=P_1$ is a generalized polymatroid as well (Theorem \ref{THgp}\ref{ENUclosed}). 
\end{proof}

By Theorem \ref{THgp}\ref{ENUinteger}, 
the generalized polymatroids $P_1$ and $P_2$ are integer. 
It then follows from Theorem \ref{THgp}\ref{ENUintersection} 
that $P_1 \cap P_2$ is an integer polyhedron. 
In the next claim, 
we show that $P_1 \cap P_2$ is nonempty, 
which certifies that $P_1 \cap P_2$ contains an integer vector. 
Denote by $\bm{1}^H$ the vector in $\RR^H$ each of whose component is one. 

\begin{claim}
\label{CLoverk}
The vector $x^*=\bm{1}^H/k$ belongs to $P_1 \cap P_2$. 
\end{claim}

\begin{proof}
Here we prove $x^* \in P_1$. 
We can prove $x^* \in P_2$
in the same manner. 

It is clear that $x^*$ satisfies \eqref{EQgpoly1}. 
We obtain \eqref{EQgpoly2} as follows: 
\begin{align*}
x^*(C) 
{}&{}= \frac{|C|}{k} = |C| - \frac{k-1}{k}|C| \\
{}&{}\le |C| - \frac{k-1}{k}g_1(C) = |C| - g_1(C) + \frac{1}{k}g_1(C) \\
{}&{}\le |C| - g_1(C) + 1, 
\end{align*}
where the two inequalities follow from \eqref{EQcolorassump}. 

If $g_1(C)=k$, 
then, 
$|C| \ge k$ follows 
from \eqref{EQcolorassump}. 
This implies $x^*(C) \ge 1$, 
and thus \eqref{EQgpoly3} is satisfied. 

Finally, 
for $v \in T$, 
\begin{align*}
(d_A^-(v) - (k-1)b(v)) - x^*(\delta_H^-(v))
{}&{}= (d_A^-(v) - (k-1)b(v)) - \frac{|\delta_H^-(v)|}{k} \\
{}&{}\ge (d_A^-(v) - (k-1)b(v)) - \frac{d_A^-(v)}{k} \\
{}&{}= (k-1)\left(\frac{d_A^-(v)}{k} - b(v)\right) \\
{}&{}\ge 0,
\end{align*}
where the latter inequality follows from 
the definition of $k$, 
i.e.,\ 
$k$ is at most \eqref{EQminT}. 
Thus, 
$x^*$ satisfies \eqref{EQgpoly4}. 
Therefore, we conclude that 
$x^* \in P_1$. 
\end{proof}

Now 
$P_1 \cap P_2$ contains an integer vector $x_1 \in \{0,1\}^H$. 
Denote the arc subset of $H$ whose characteristic vector is $x_1$ by $H_1$.  
By induction, 
we obtain a partition $\{H_1,\ldots, H_k\}$ of 
$H=A[S,T]$ 
satisfying 
\begin{align}
\label{EQcolored}
&{}|\{ j \in [k] \colon C \cap H_j \neq \emptyset \}| \ge g_i(C) 
\quad \mbox{for each $C \in \C_i$ and each $i=1,2$}, \\
\label{EQatmostkT}
&{}
d_{H_j}^-(v)\le d_A^-(v) - (k-1)b(v) \le b(v) \quad \mbox{for each $v \in T$ and $j \in [k]$}, \\
\label{EQatmostkS}
&{}
d^+_{H_j}(u)
\le d_A^+(u) - (k-1)b(v) \le b(v)\quad \mbox{for each $u \in S$ and $j \in [k]$} .
\end{align}

We complete the proof by showing that $D[T]$ has disjoint $b$-branchings $B_1,\ldots, B_k$ and 
$D[S]$ has disjoint $b$-cobranchings $B^*_1,\ldots, B^*_k$ such that 
$B_j^* \cup H_j \cup B_j$ ($j\in [k]$) is a $b$-bibranching. 
Let $U$ be an arbitrary nonempty subset of $T$. 
It follows from the definition \eqref{EQg1} of $g_1$ that 
$g_1(\delta_H^-(U)) \ge k - d_{A[T]}^-(U)$. 
Combined with 
\eqref{EQcolored}, 
this implies that 
\begin{align}
\label{EQpink}
|\{ j \in [k] \colon \delta_{H_j}(U) \neq \emptyset \}| 
\ge g_1(\delta_H^-(U)) \ge k - d_{A[T]}^-(U).
\end{align}
For $j \in [k]$, 
define $b_j \in \ZZ_+^T$ by 
\begin{align*}
b_j(v) = b(v) - d_{H_j}^-(v) \quad (v \in T). 
\end{align*}
By \eqref{EQatmostkT}, 
we have that $b_j(v) \ge 0$ for each $j\in [k]$ and each $v \in T$. 
It then follows that 
\begin{align*}
\sum_{j=1}^k b_j(v) 
= 
\sum_{j=1}^k (b(v) - d_{H_j}^-(v)) 
=
k \cdot b(v) - d_H^-(v)
\le 
d_A^-(v) - d_H^-(v)  
=
d_{A[T]}^-(v), 
\end{align*}
where the inequality follows from \eqref{EQminT}. 

Moreover, 
by \eqref{EQpink}, 
\begin{align*}
d_{A[T]}^-(U) 
{}&{}\ge k - |\{ j \in [k] \colon \delta_{H_j}^-(U) \neq \emptyset \}|\\
{}&{}= |\{ j \in [k] \colon \delta_{H_j}^-(U)   = \emptyset \}|\\
{}&{}= |\{ j \in [k] \colon b_j(U) = b(U) \neq 0 \}|.
\end{align*}
Thus, 
by Theorem \ref{THMbbpacking}, 
$D[T]$ has disjoint $b$-branchings $B_1, \ldots, B_k$ such that $d_{B_j}^-=b_j$ for each $j \in [k]$. 
In the same manner, 
we can also see that $D[S]$ has disjoint $b$-cobranchings $B^*_1, \ldots, B^*_k$ such that $d_{B_j}^+=b - d^+_{H_j}$ for each $j \in [k]$. 
We now have that 
$B^*_j \cup H_j \cup B_j$ ($j \in [k]$) are 
disjoint $b$-bibranchings in $D$. 
\end{proof}

The integer decomposition property of the $b$-bibranching polytope 
is a direct consequence of Theorem \ref{THMbbbpacking}. 
A polytope $P \subseteq \RR^A$ has the \emph{integer decomposition property} if, 
for an arbitrary positive integer $k$, 
an arbitrary integer vector $x \in kP$ can be represented by 
the sum of $k$ integer vectors in $P$, 
where $kP=\{ x \in \RR^A \colon x = kx', x' \in P \}$. 

\begin{corollary}
The $b$-bibranching polytope has the integer decomposition property.
\end{corollary}

\begin{proof}
Let $P \subseteq \RR^A$ denote the $b$-bibranching polytope 
and let $k$ be an arbitrary positive integer. 
By Corollary \ref{CORbbbpoly}, 
$kP$ is described as 
\begin{alignat}{2}
\label{EQkP1}
&{}x(\delta^-(v)) \ge k\cdot b(v) \quad{}&&{} \mbox{for each $v \in T$}, \\
&{}x(\delta^+(v)) \ge k\cdot b(v) {}&&{} \mbox{for each $v \in S$}, \\
&{}x(C) \ge k {}&&{} \mbox{for each bicut $C$}, \\
\label{EQkP4}
&{} 0 \le x(a) \le k {}&&{} \mbox{for each $a \in A$}. 
\end{alignat}
Let $x$ be an integer vector in $kP$ 
and 
let $A_x$ be a multiset of arcs such that $a \in A$ is contained in $A_x$ with multiplicity $x(a)$. 
Since $x$ satisfies \eqref{EQkP1}--\eqref{EQkP4}, 
it follows from Theorem \ref{THMbbbpacking} that 
the digraph $(V,A_x)$ contains $k$ disjoint $b$-bibranchings. 
Since a superset of a $b$-bibranching is a $b$-bibranching, 
$A_x$ can be partitioned into $k$ $b$-bibranchings. 
We thus conclude that 
$P$ has the integer decomposition property. 
\end{proof}

\section{M-convex submodular flow formulation}
\label{SECm}

In this section, 
we present an $\mathrm{M}^\natural$-convex submodular flow formulation of the shortest $b$-bibranching problem. 
We remark that this formulation implies a combinatorial polynomial algorithm. 

We first show an extension of Theorem \ref{THMMconv}: an $\mathrm{M}^\natural$-convex function is derived from $b$-branchings. 
Let $D=(V,A)$ be a digraph and 
$w \in \RR_+\sp{A}$ represent the arc weights. 

Define a function $g \colon \ZZ^V \to \overline{\ZZ}$  
by 
\begin{align}
\label{EQdomg}
&{}\dom g = \{ x \in \ZZ_+^V \colon \mbox{$D$ has a $b$-branching $B$ with $d_B^- + x \ge b$} \}, \\
\label{EQvalg}
&{}g(x) = 
\begin{cases}
\min\{ w(B) \colon \mbox{$B$ is a $b$-branching, $d_B^- + x \ge b$} \} & (x \in \dom g), \\
+\infty & (x \not \in \dom g). 
\end{cases}
\end{align}

\begin{theorem}
\label{THMg}
The function $g \colon \ZZ^V \to \overline{\ZZ}$  defined by 
\eqref{EQdomg} and \eqref{EQvalg} is an $\mathrm{M}^\natural$-convex function. 
\end{theorem}

To prove Theorem \ref{THMg}, 
we define another function $f \colon \ZZ^V \to \overline{\ZZ}$, 
which is derived from $b$-branchings more directly. 
Define $f \colon \ZZ^V \to \overline{\ZZ}$ by 
\begin{align}
\label{EQdomf}
&{}\dom f = \{x \in \ZZ_+^V \colon \mbox{$D$ has a $b$-branching $B$ with $d_B^- + x = b$}\}, \\ 
\label{EQvalf}
&{}f(x) = 
\begin{cases}
\min\{ w(B) \colon \mbox{$B$ is a $b$-branching, $d_B^- + x = b$} \} & (x \in \dom f), \\
+\infty & (x \not \in \dom f). 
\end{cases}
\end{align}

\begin{lemma}
\label{LEMf}
The function $f \colon \ZZ^V \to \overline{\ZZ}$  defined by 
\eqref{EQdomf} and \eqref{EQvalf} is an $\mathrm{M}^\natural$-convex function. 
\end{lemma}

Lemma \ref{LEMf} follows from the exchange property of $b$-branchings (Lemma \ref{LEMexc}). 
Call a strong component $X$ in $D$ a \emph{source component} 
if $\delta^-_A(X) = \emptyset$.

\begin{lemma}
\label{LEMtwopartition}
Let $D=(V,A)$ be a digraph, 
and 
$b\in \ZZ_{++}^V$ be a positive integer valued vector on $V$ such that 
$A$ can be partitioned into two $b$-branchings $B_1,B_2 \subseteq A$. 
Let $b_1',b_2' \in \ZZ^V$ satisfy 
\begin{align}
\label{EQtwopartition}
{}&{}b_1' + b_2' = d_A^-, \\
{}&{}b_1' \le b, \quad b_2' \le b. 
\end{align}
Then, 
$A$ can be partitioned into two $b$-branchings $B_1',B_2' \subseteq A$ 
with 
$d_{B'_1}^- = b_1'$ and 
$d_{B'_2}^- = b_2'$ 
if and only if 
\begin{align}
\label{EQexcsource}
b_1'(X) \lneqq b(X), \quad b_2'(X) \lneqq b(X) 
\quad \mbox{for each source component $X$ in $D$.}
\end{align}
\end{lemma}

\begin{proof}
Necessity is easy, 
and here we prove sufficiency 
by Theorem \ref{THMbbpacking}. 
Since 
\eqref{EQpackingdeg} follows from \eqref{EQtwopartition},  
it suffices to show \eqref{EQpackingcut}, 
i.e., 
\begin{align}
\label{EQexccut}
d_A^-(U) \ge | \{i \in \{1,2\} \colon b'_i(U) = b(U) \neq 0\}| \ \mbox{for each nonempty subset $U$ of $V$.}
\end{align}
Suppose to the contrary that \eqref{EQexccut} does not hold for some nonempty set $U \subseteq V$. 
Note that $b(U) \neq 0$ holds for every $U\neq \emptyset$ since $b \in \ZZ_{++}^V$.

If $| \{i \in \{1,2\} \colon b'_i(U) = b(U) \neq 0\}| \le 1$, 
then 
it should hold that 
\begin{align}
\label{EQzeroindeg}
{}&{}d_A^-(U) = 0 ,\\
\label{EQexcRHS1}
{}&{}| \{i \in \{1,2\} \colon b'_i(U) = b(U)\neq 0\}| = 1. 
\end{align}
By \eqref{EQzeroindeg}, 
$U$ contains a source component $X$. 
Then, 
by \eqref{EQexcsource}, 
it holds that 
$b_1'(U) \lneqq b(U)$ and 
$b_2'(U) \lneqq b(U)$, 
contradicting to \eqref{EQexcRHS1}.

If $| \{i \in \{1,2\} \colon b'_i(U) = b(U)\neq 0\}| = 2$, 
it follows that 
$\sum_{v\in U}d_{B_1}^-(v) = \sum_{v\in U}d_{B_2}^-(v) = b(U)$. 
This implies that 
$\delta_{B_1}^-(U) \neq \emptyset$ and 
$\delta_{B_2}^-(U) \neq \emptyset$. 
We thus obtain $d_A^-(U) \ge 2$, 
contradicting that \eqref{EQexccut} does not hold. 
Therefore, 
we have shown \eqref{EQexccut}. 
\end{proof}

\begin{lemma}
\label{LEMexc}
Let $D=(V,A)$ be a digraph 
and 
$b\colon V \to \ZZ_{++}$. 
Let $B_1, B_2 \subseteq A$ be $b$-branchings, 
and let a vertex $s \in V$ satisfy 
$d_{B_1}^-(s) \lneqq d_{B_2}^-(s)$. 
Then, 
$D$ has $b$-branchings $B_1', B_2' \subseteq A$ satisfying the following: 
\begin{enumerate}
\item
	\label{ENUcup}
	$B_1' \cup B_2' = B_1 \cup B_2$, 
\item
	\label{ENUcap}
	$B_1' \cap B_2' = B_1 \cap B_2$, 
\item
	\label{ENUstep}
	at least one of {\rm (a)} and {\rm (b)} below holds: 
	\begin{enumerate}
	\item
	\label{ENU1step}
	$d_{B_1'}^- =d_{B_1}^- + \chi_s$ and $d_{B_2'}^- =d_{B_2}^- - \chi_s$, 
	\item
	\label{ENU2step}
		there exists $t \in V$ such that 
		$d_{B_2}^-(t) \lneqq d_{B_1'}^-(t)$, 
		$d_{B_1'}^- =d_{B_1}^- + \chi_s - \chi_t$, and $d_{B_2'}^- =d_{B_2}^- - \chi_s + \chi_t$. 
	\end{enumerate}
\end{enumerate}
\end{lemma}

\begin{proof}
Let $X$ be a strong component in $D$ containing the vertex $s$. 
Suppose that $X$ is a source component in $D$ and $\sum_{v \in X}d_{B_1}^-(v) = b(X)-1$. 
Since $d_{B_1}^-(s) \lneqq d_{B_2}^-(s)$, 
this implies that 
\begin{align*}
{}&{}d_{B_1}^-(v) = 
\begin{cases}
b(v)    & (v \in X \setminus \{s\}), \\
b(v) -1 & (v=s),
\end{cases}
&{}&{}d_{B_2}^-(s) = b(s). 
\end{align*}
Then, 
since $B_2$ is a $b$-branching and $X$ is a source component, 
there exists a vertex $t \in X \setminus \{s\}$ such that 
$d_{B_2}^-(t) \lneqq b(v) = d_{B_1}(t)$. 
Now define $b_1',b'_2 \colon V \to \ZZ_{++}$ by 
$b_1' = d_{B_1}^- + \chi_s - \chi_t$ 
and  
$b_2' = d_{B_2}^- - \chi_s + \chi_t$. 
It then follows from Lemma \ref{LEMtwopartition} that 
$D$ has $b$-branchings $B_1'$ and $B_2'$ satisfying \ref{ENUcup}, \ref{ENUcap}, and \ref{ENU2step}.

Suppose otherwise. 
Then define $b_1',b'_2 \colon V \to \ZZ_{++}$ by 
$b_1' = d_{B_1}^- + \chi_s$ and  
$b_2' = d_{B_2}^- - \chi_s$. 
Again 
from Lemma \ref{LEMtwopartition}, 
it follows that $D$ has $b$-branchings $B_1'$ and $B_2'$ satisfying \ref{ENUcup}, \ref{ENUcap}, and \ref{ENU1step}.
\end{proof}

Lemma \ref{LEMf} directly follows from Lemma \ref{LEMexc}. 

\begin{proof}[Proof of Lemma \ref{LEMf}]
Let $x,y \in \dom f$ and $s \in \suppp (x-y)$. 
Denote $b$-branchings attaining $f(x)$ and $f(y)$ by 
$B_x$ and $B_y$, 
respectively. 
That is, 
$d^-_{B_x} +x=b$, 
$w(B_x) = f(x)$, 
$d^-_{B_y} +y=b$, 
$w(B_y) = f(y)$. 
It follows from $s \in \suppp (x-y)$ that 
$d^-_{B_x}(s) < d^-_{B_y}(s)$. 
Now apply Lemma \ref{LEMexc} to $B_x$, $B_y$, and $s$ 
to obtain $b$-branchings $B_x'$ and $B_y'$. 
If \ref{ENU1step} in Lemma \ref{LEMexc} holds for $B_x'$ and $B_y'$, 
then 
we obtain \eqref{EQM1}: 
$$f(x - \chi_s) + f(y + \chi_s) \le w(B_x') + w(B_y') = w(B_x) + w(B_y)=f(x)+f(y).$$
If \ref{ENU2step} in Lemma \ref{LEMexc} holds for $B_x'$ and $B_y'$, 
then 
we obtain \eqref{EQM2}: 
$$f(x - \chi_s+ \chi_t) + f(y + \chi_s - \chi_t) \le w(B_x') + w(B_y') = w(B_x) + w(B_y)=f(x)+f(y).$$ 
We thus conclude that $f$ satisfies the exchange property of a $\mathrm{M}^\natural$-convex function. 
\end{proof}

We are now ready to prove Theorem \ref{THMg}. 

\begin{proof}[Proof of Theorem \ref{THMg}]
Let $x,y \in \dom g$ and 
$u \in \suppp(x-y)$. 
We have three cases: 
(i) $y(u) \ge b(u)$; 
(ii) $x(u) \ge b(u) + 1$ and $y(u) \le b(u)-1$; 
and 
(iii) $y(u)<x(u) \le b(u)$. 

\paragraph{Case (i).}
If $y(u) \ge b(u)$, 
then $g(x - \chi_u)=g(x)$ and $g(y+\chi_u) = g(y)$ follow. 

\paragraph{Case (ii).}
If $x(u) \ge b(u) + 1$, 
then $g(x - \chi_u)=g(x)$. 
Moreover, 
$y(u) \leq b(u)-1$ implies that $g(y+\chi_u) \le g(y)$. 
This is explained as follows. 
Let $B_y \subseteq A$ be a $b$-branching attaining $g(y)$. 
That is, 
$B_y$ is a $b$-branching satisfying 
$d^-_{B_y} +y \ge b $ and $w(B_y) = g(y)$. 
It then holds that 
$d^-_{B_y}(u) \ge b(u) -y(u) \ge 1$, 
and hence there exists an arc $a \in B_y$ with $\partial^-a=u$. 
Now define $B_y' = B_y \setminus \{a\}$, 
and we obtain 
$d_{B_y'}^- =d_{B_y'}^- - \chi_u \ge (b-y) - \chi_u = b-(y+\chi_u)$. 
This implies that 
$g(y+\chi_u) \le w(B_y') =w(B_y)-w(a) \le w(B_y) = g(y)$. 
We thus conclude that 
$g(x) + g(y)\ge g(x-\chi_u)+g(y+\chi_u) $. 

\paragraph{Case (iii).}
Let $B_x$ and $B_y$ be $b$-branchings attaining $g(x)$ and $g(y)$, 
respectively. 
In Lemma \ref{LEMexc}, 
put to $B_1=B_y$, $B_2=B_x$, and $s=u$. 
We then obtain $b$-branchings $B_1'$ and $B_2'$ 
satisfying (i), (ii), and (iii)(a); 
or (i), (ii), and (iii)(b). 
In the former case, 
(iii)(a) implies that 
$g(y+\chi_s) \le w(B_1')$ and 
$g(x-\chi_s) \le w(B_2')$. 
It also follows from (i) and (ii) that $w(B_1')+w(B_2') = w(B_1) + w(B_2) = g(y) + g(x)$. 
We thus obtain $g(y+\chi_s) + g(x-\chi_s) \le g(y) + g(x)$. 
In the latter case, 
(iii)(b) implies that 
there exists $t \in \suppp (y-x)$ satisfying 
$g(y+\chi_s - \chi_t) \le w(B_1')$ and 
$g(x-\chi_s+ \chi_t) \le w(B_2')$. 
We similarly obtain 
$g(y+\chi_s-\chi_t) + g(x-\chi_s+\chi_t) \le g(y) + g(x)$. 
\end{proof}

Our $\mathrm{M}^\natural$-convex submodular flow formulation of the shortest $b$-bibranching problem can be 
derived immediately from Theorem \ref{THMg}. 
Let $D=(V,A)$ be a digraph, 
$w \in \RR_+^A$ be a vector representing the arc weights, 
$\{S,T\}$ be a partition of $V$, where $S,T \neq \emptyset$, 
and 
$b\in \ZZ_{++}^V$ be a positive integer vector on $V$. 
Define two functions $g_T \colon \ZZ^T \to \overline{\ZZ}$ and 
$g_S \colon \ZZ^S \to \overline{\ZZ}$ by 
\begin{align*}
&{}\dom g_T = \{ x \in \ZZ_+^T \colon \mbox{$D[T]$ has a $b|_T$-branching $B$ with $d_B^- +x\ge b|_T$} \}, \\
&{}g_T(x) = 
\begin{cases}
\min\{ w(B) \colon \mbox{$B$ is a $b|_T$-branching in $D[T]$, $d_B^- + x \ge b|_T $} \} & (x \in \dom g_T), \\
+\infty & (x \not \in \dom g), 
\end{cases}\\
&{}\dom g_S = \{ x \in \ZZ_+^S \colon \mbox{$D[S]$ has a $b|_S$-cobranching $B$ with $d_B^++x\ge b|_S$} \}, \\
&{}g_S(x) = 
\begin{cases}
\min\{ w(B) \colon \mbox{$B$ is a $b|_S$-cobranching in $D[S]$, $d_B^++x \ge b|_S $} \} & (x \in \dom g_S), \\
+\infty & (x \not \in \dom g). 
\end{cases}
\end{align*}
Then it follows from Theorem \ref{THMg} that $g_T$ and $g_S$ are $\mathrm{M}^\natural$-convex functions. 
Now 
the shortest $b$-bibranching problem can be formulated as 
the following $\mathrm{M}^\natural$-convex submodular flow problem in variable $\xi \in \ZZ^{A[S,T]}$: 
\begin{alignat*}{2}
&{}\mbox{minimize} \quad 	{}&&{}\sum_{a \in A[S,T]}w(a)\xi(a) + g_S(\partial^+ \xi) + g_T(\partial^- \xi) \\
&{}\mbox{subject to} \quad 	{}&&{}
0 \le \xi(a) \le 1 \quad \mbox{for each $a \in A[S,T]$}, \\
&{} {}&&{}\partial^+\xi \in \dom g_S , \\
&{} {}&&{}\partial^-\xi \in \dom g_T.
\end{alignat*}

This $\mathrm{M}^\natural$-convex submodular flow formulation provides 
a combinatorial polynomial algorithm 
for the shortest $b$-bibranching problem in the following manner. 
The $\mathrm{M}^\natural$-convex submodular flow problem can be solved 
by polynomially many calls of an oracle for computing the $\mathrm{M}^\natural$-convex function values \cite{IMM05,IS02}. 
In our formulation, 
this computation amounts to computing the minimum weight of a $b$-branching with prescribed indegree, 
which can be done by using a combinatorial algorithm for the longest $b$-branching \cite{KKT18}. 

\section{Conclusion}
\label{SECconcl}

In this paper, 
we have proposed the $b$-bibranching problem, 
which is a new framework of tractable combinatorial optimization problem. 
We have proved its polynomial solvability in two ways. 
One is based on the linear ming formulation with total dual integrality. 
The other is based on the $\mathrm{M}^\natural$-convex submodular flow formulation. 
We have also presented a min-max theorem for packing disjoint $b$-bibranchings, 
which is an extension of Edmonds' disjoint branchings theorem. 
With this packing theorem, 
we have proved that the $b$-bibranching polytope is a new example of a polytope with integer decomposition property. 

While this research is mainly motivated from theoretical interest, 
it would have a potential to be applied to more practical problems, 
such as 
the evacuation and communication network design problems. 

\section*{Acknowledgements}

The author thanks Kazuo Murota for helpful comments. 
This work is partially supported by 
JST CREST Grant Number JPMJCR1402,  
JSPS KAKENHI Grant Numbers 
JP16K16012, 
JP25280004, 
JP26280001, 
JP26280004, 
Japan.

\end{document}